\documentclass[a4paper,11pt]{article}
\usepackage[T1]{fontenc}
\usepackage[utf8]{inputenc}
\usepackage{lmodern}
\usepackage[a4paper,margin=2.6cm]{geometry}
\usepackage{amsmath,amssymb,amsthm,mathtools}
\usepackage{enumitem}
\usepackage{microtype}
\usepackage{float}
\usepackage{tikz}
\usepackage{pgfplots}
\pgfplotsset{compat=1.18}
\usepackage{xcolor}
\usepackage{listings}
\usepackage[hidelinks]{hyperref}

\newtheorem{theorem}{Theorem}
\newtheorem{proposition}{Proposition}

\theoremstyle{definition}
\newtheorem{definition}{Definition}
\theoremstyle{remark}
\newtheorem{remark}{Remark}

\newcommand{\R}{\mathbb{R}}
\newcommand{\1}{\mathbf{1}}
\newcommand{\SP}{\mathrm{SP}}

\title{\sc Exact Dynamic Programming for Solow--Polasky Diversity Subset Selection on Lines and Staircases}
\author{
Michael T. M. Emmerich\thanks{ORCID: \href{https://orcid.org/0000-0002-7342-2090}{0000-0002-7342-2090}. Email: \href{mailto:michael.t.m.emmerich@jyu.fi}{michael.t.m.emmerich@jyu.fi}.}\\
Faculty of Information Technology\\
University of Jyv\"askyl\"a\\
Jyv\"askyl\"a, Finland
}
\date{}

\begin{document}
\maketitle

\begin{abstract}
This paper studies exact fixed-cardinality Solow--Polasky diversity subset selection on
ordered finite $\ell_1$ point sets, with monotone biobjective Pareto fronts and their
higher-dimensional staircase analogues as central applications. Solow--Polasky
diversity was introduced in biodiversity conservation, whereas the same
inverse-matrix expression appears in metric geometry as magnitude: for a finite
metric space $(X,d)$ with exponential similarity matrix
$Z_{ij}=e^{-q d(x_i,x_j)}$, the quantity $\1^\top Z^{-1}\1$ is the magnitude of
the scaled finite metric space $(X,qd)$ whenever the weighting is defined by the
inverse matrix. Thus, in this finite exponential-kernel setting,
Solow--Polasky diversity and magnitude are mathematically the same object viewed
through different motivations. Building on the linear-chain magnitude formula of
Leinster and Willerton, the paper gives a detailed proof of the scaled
consecutive-gap identity
$
\SP(X)=1+\sum_r \tanh(qg_r/2),
$
where the $g_r$ are the gaps between consecutive selected points. It then proves
an exact Bellman-recursion theorem for maximizing this value over all subsets of
a prescribed cardinality, yielding an $O(kn^2)$ dynamic program for an ordered
$n$-point candidate set and subset size $k$. Finally, the paper proves ordered
$\ell_1$ reductions showing that the same algorithm applies to monotone
biobjective Pareto-front approximations and, more generally, to finite
coordinatewise monotone $\ell_1$ staircases in $\R^d$. These are precisely the
ordered $\ell_1$ chains for which the $\ell_1$-distance becomes a line metric
along the chosen order, so the one-dimensional dynamic program applies without
modification. 

\end{abstract}

\section{Introduction and notation}\label{sec:intro}
Diversity subset selection is the problem of choosing, from a finite candidate
set, a prescribed number of representatives that remain sufficiently separated
with respect to a given distance or similarity model. Such a question is natural
in biodiversity conservation, where one wants to quantify the effective
diversity of a collection of species, and also in multiobjective optimization,
where a finite approximation of a Pareto front should not collapse to several
nearly identical alternatives. In full generality, fixed-cardinality
Solow--Polasky subset selection is computationally hard in general metric spaces~\cite{EmmerichPereverdievaDeutz2026Metric} and remains hard even for finite point sets in the Euclidean plane~\cite{EmmerichPereverdievaDeutz2026Plane}. The point of the
present paper is more modest and more specific: it isolates an ordered tractable
case in which the metric is, in effect, a line metric. This case includes
ordered points on the real line, monotone biobjective Pareto-front
approximations under the $\ell_1$ metric\footnote{The $\ell_1$-distance is also commonly called the Manhattan distance; I use $\ell_1$ terminology throughout.}, and their higher-dimensional
staircase analogues.

Let $q>0$ and let $(Y,d)$ be a finite metric space. For a finite subset
$A=\{y_1,\dots,y_m\}\subseteq Y$, define the exponential similarity matrix
\[
Z(A)=(Z_{ij})_{i,j=1}^m,
\qquad
Z_{ij}=e^{-q d(y_i,y_j)}.
\]
Whenever $Z(A)$ is invertible, the \emph{Solow--Polasky diversity} of $A$ is
\[
\SP(A):=\1^\top Z(A)^{-1}\1.
\]
In its original interpretation, $Z$ is a biological similarity matrix and
$\SP(A)$ measures an effective number of mutually dissimilar species or
alternatives. In metric geometry, the same inverse-matrix expression is known as
the \emph{magnitude} of a finite metric space. With the exponential kernel used
above, $\SP(A)$ is precisely the magnitude of the scaled finite metric space
$(A,qd)$. Thus, in the finite exponential-kernel setting considered here,
Solow--Polasky diversity and finite magnitude are not merely analogous objects;
they are the same mathematical quantity viewed through two different
motivations. This connection has also been emphasized by Huntsman~\cite{Huntsman2023}.

The essential simplification used in this paper occurs on ordered chains. To fix notation, first
consider points on the line,
\[
X=\{x_1,\dots,x_k\}\subseteq \R,
\qquad
x_1<\cdots <x_k,
\]
with distance $d(x_i,x_j)=|x_i-x_j|$. For consecutive points write
\[
g_r:=x_{r+1}-x_r>0,
\qquad r=1,\dots,k-1.
\]
Leinster and Willerton proved, in their study of magnitude of finite metric
spaces and finite approximations of compact Euclidean sets, that a linear chain
with consecutive metric gaps $d_1,\dots,d_{k-1}$ has magnitude
$1+\sum_{r=1}^{k-1}\tanh(d_r/2)$~\cite[Theorem~4 and
Proposition~6]{LeinsterWillerton2009}. After the scaling $d_r=qg_r$, this gives
\[
\SP(X)=1+\sum_{r=1}^{k-1}\tanh(qg_r/2).
\]
Consequently, the hyperbolic-tangent identity itself is not claimed here as a
new formula. The contribution is instead to put this identity to work for
fixed-cardinality Solow--Polasky subset selection, to state the resulting exact
Bellman recursion explicitly, and to identify the ordered $\ell_1$ geometries in
which the same recursion applies without modification.

For an ordered candidate set of size $n$, the preceding formula turns the
selection of $k$ representatives into a dynamic programming problem. The
contribution of a newly selected point depends only on the gap to the previously
selected point. Hence a Bellman recursion indexed by the number of selected
points and by the last selected candidate gives an exact $O(kn^2)$ algorithm.
Let us remark here that same ordered-chain structure also gives an exact bottleneck recursion for
the Minimum Pairwise Distance (MPD) objective\footnote{In multiobjective diversity optimization this criterion is also often referred to as Max--Min Diversity.}, although that
objective is not additive.

The one-dimensional formulation is useful beyond literal subsets of the real
line. For a monotone biobjective Pareto-front approximation measured in the
$\ell_1$ metric, ordering the points along the front makes every pairwise
$\ell_1$-distance equal to the length of the corresponding interval on an
associated line. Thus the line dynamic program applies verbatim. The same idea
extends to higher dimensions via finite coordinatewise monotone $\ell_1$ staircases in $\R^d$: after
possible reversals of coordinate axes, a single ordering must sort all
coordinate projections. This provides the ordered $\ell_1$-chain condition
under which the $\ell_1$ metric becomes a line metric along the chosen order, as we will prove in this paper. 

The paper is organized as follows. Section~\ref{sec:related-work} reviews the
origin of Solow--Polasky diversity, its relation to magnitude, known positive
and negative complexity results, and its use in multiobjective diversity
optimization. Section~\ref{sec:line-formula} proves the consecutive-gap formula
on the ordered line in the notation used here. Section~\ref{sec:dp} proves the
Bellman-recursion theorem for fixed-cardinality subset selection and records the
corresponding MPD recursion. Section~\ref{sec:l1-pareto} proves the
ordered $\ell_1$ planar reduction for monotone biobjective Pareto fronts and
illustrates it with two examples: one small worked instance and one randomly
sampled front. Section~\ref{sec:l1-staircases} gives the higher-dimensional
monotone-staircase characterization.

\section{Related work}\label{sec:related-work}
Solow--Polasky diversity was introduced in biodiversity conservation as a way to
measure the effective diversity of a collection of species from pairwise
similarities or distances, motivated by statistical comparison arguments rather
than only by species counts~\cite{SolowPolaskyBroadus1993}. In the finite metric
setting used here, with similarity matrix
\[
Z_{ij}=e^{-q d(x_i,x_j)},
\]
the same inverse-matrix expression $\1^\top Z^{-1}\1$ is also the magnitude of
the scaled metric space $(X,q d)$, whenever the inverse expression is
well-defined. Thus Solow--Polasky diversity is not merely analogous to magnitude
in this setting; it is the finite exponential-kernel instance of the same
mathematical quantity. This connects the biodiversity-conservation indicator to
the broader magnitude and diversity framework of metric geometry and
biology~\cite{LeinsterMeckes2016}.

The literature contains both positive and negative algorithmic results, but for
closely related rather than identical optimization problems. Leinster and Meckes
studied maximization of diversity over distributions on a fixed similarity
matrix. Their general finite algorithm can be exponential in the number of
species, but they also identify important polynomial-time cases, including
ultrametric matrices and matrices sufficiently close to the identity
matrix~\cite{LeinsterMeckes2016}. These results are compatible with the present
work: they concern diversity-maximizing distributions for a given set, whereas
the present paper studies fixed-cardinality subset selection from a candidate set.

For the fixed-cardinality subset-selection problem, recent results show that the
general problem is computationally hard. Emmerich, Pereverdieva, and Deutz prove
NP-hardness in general metric spaces for every fixed kernel parameter
$q>0$~\cite{EmmerichPereverdievaDeutz2026Metric}. A follow-up result strengthens
this by proving NP-hardness even for finite point sets in the Euclidean plane,
again for every fixed positive kernel parameter~\cite{EmmerichPereverdievaDeutz2026Plane}.
The ordered one-dimensional and ordered $\ell_1$ Pareto-front cases treated in
this note should therefore be read as structured tractable cases rather than as
representative of the unrestricted problem.

Solow--Polasky diversity has also been used in evolutionary multiobjective
optimization and diversity optimization. Ulrich, Bader, and Thiele introduced an
indicator-based approach in which diversity of solution sets is optimized in the
context of multiobjective search~\cite{UlrichBaderThiele2010}. More recently,
Pereverdieva et al. compared Solow--Polasky diversity with other indicators for
multiobjective diversity optimization, including MPD and Riesz $s$-Energy~\cite{Uribe2024Riesz},
and analyzed properties relevant for subset selection and indicator-based
evolutionary algorithms~\cite{PereverdievaEtAl2025}.

Studies on related diversity measures have also considered dynamic-programming
approaches. MPD is closely related to Riesz $s$-energy: for a fixed finite set,
the transformed Riesz quantity
\(\left(\sum_{r<s} d(x_r,x_s)^{-s}\right)^{-1/s}\)
converges to the MPD value as \(s\to\infty\). Thus, MPD can be viewed as the
limiting bottleneck case of Riesz $s$-energy subset selection. This comparison is
useful in light of the ordered-point Riesz dynamic-programming study of
Emmerich~\cite{Emmerich2025RieszDP}. The present note is motivated by the same
Pareto-front approximation setting, but focuses on the special ordered
\(\ell_1\) geometry in which an exact dynamic program is available.

\section{The consecutive-gap formula on the line}\label{sec:line-formula}

This section records the line formula in the notation needed for the dynamic program and provides a new, detailed, derivation of the scaled version of the Leinster-Willerton expression \cite{LeinsterWillerton2009}. It starts from a simple but important  observation related to the fact that the exponential kernel function $e^{-x}$ is a function with the 'nice' property $$f(x+y)=f(x)f(y).$$

For convenience write
\[
a_r:=e^{-qg_r}\in(0,1),
\qquad r=1,\dots,k-1.
\]
If $i<j$, then
\[
|x_j-x_i|=g_i+g_{i+1}+\cdots +g_{j-1},
\]
and therefore
\begin{equation}
\label{eq:product-form}
Z_{ij}=e^{-q(x_j-x_i)}=a_i a_{i+1}\cdots a_{j-1}.
\end{equation}
Thus every similarity entry is already determined by the consecutive gaps.
The following theorem shows that the same is true for the full
Solow--Polasky value. It is the scaled Solow--Polasky form of the linear
metric-space formula of Leinster and Willerton~\cite[Theorem~4]{LeinsterWillerton2009}.

\begin{theorem}[Consecutive-gap formula; scaled Leinster--Willerton form]
\label{thm:gap-formula}
Let $x_1<\cdots <x_k$ and let $g_r=x_{r+1}-x_r$. Then
\[
\SP(X)
=
1+\sum_{r=1}^{k-1}\tanh\!\left(\frac{qg_r}{2}\right).
\]
Equivalently, with $a_r=e^{-qg_r}$,
\[
\SP(X)=1+\sum_{r=1}^{k-1}\frac{1-a_r}{1+a_r}.
\]
In particular, $\SP(X)$ depends only on the consecutive gaps $g_1,\dots,g_{k-1}$.
\end{theorem}

\begin{proof}
Let $w=(w_1,\dots,w_k)^\top$ be the unique solution of
\[
Zw=\1.
\]
Then
\[
\SP(X)=\1^\top Z^{-1}\1=\1^\top w=\sum_{i=1}^k w_i.
\]
So it suffices to determine the weights $w_i$ explicitly.

For each $i\in\{1,\dots,k\}$, let $E_i$ denote the $i$-th row equation of
$Zw=\1$:
\begin{equation}
\label{eq:Ei}
\sum_{j<i} Z_{ji}w_j + w_i + \sum_{j>i} Z_{ij}w_j = 1.
\end{equation}
We compute the weights from these equations.

\medskip
\noindent
\textbf{Step 1: endpoint weights.}
From \eqref{eq:product-form}, the first two row equations are
\[
E_1:\quad
w_1+a_1w_2+a_1a_2w_3+\cdots+a_1a_2\cdots a_{k-1}w_k=1,
\]
\[
E_2:\quad
a_1w_1+w_2+a_2w_3+\cdots+a_2a_3\cdots a_{k-1}w_k=1.
\]
Subtract $a_1E_2$ from $E_1$. Every term except the $w_1$-term cancels:
\[
(1-a_1^2)w_1=1-a_1.
\]
Hence
\begin{equation}
\label{eq:w1}
 w_1=\frac{1}{1+a_1}.
\end{equation}
By symmetry, comparing the last two row equations yields
\begin{equation}
\label{eq:wk}
 w_k=\frac{1}{1+a_{k-1}}.
\end{equation}

\medskip
\noindent
\textbf{Step 2: two auxiliary identities for interior indices.}
Fix $i\in\{2,\dots,k-1\}$. We first compare $E_i$ and $E_{i-1}$. By
\eqref{eq:product-form}, for $j<i-1$ one has
\[
Z_{ji}=a_{i-1}Z_{j,i-1},
\]
while for $j>i$ one has
\[
Z_{i-1,j}=a_{i-1}Z_{ij}.
\]
Therefore subtracting $a_{i-1}E_{i-1}$ from $E_i$ gives
\[
(1-a_{i-1}^2)\left(w_i+\sum_{j>i} Z_{ij}w_j\right)=1-a_{i-1}.
\]
Since $a_{i-1}\in(0,1)$, division by $1-a_{i-1}^2=(1-a_{i-1})(1+a_{i-1})$ yields
\begin{equation}
\label{eq:Bi}
 w_i+\sum_{j>i} Z_{ij}w_j=\frac{1}{1+a_{i-1}}.
\end{equation}

Next compare $E_i$ and $E_{i+1}$. For $j<i$ one has
\[
Z_{j,i+1}=a_i Z_{ji},
\]
and for $j>i+1$ one has
\[
Z_{i,j}=a_i Z_{i+1,j}.
\]
Thus subtracting $a_iE_{i+1}$ from $E_i$ gives
\[
(1-a_i^2)\left(\sum_{j<i} Z_{ji}w_j+w_i\right)=1-a_i,
\]
so that
\begin{equation}
\label{eq:Ai}
 \sum_{j<i} Z_{ji}w_j+w_i=\frac{1}{1+a_i}.
\end{equation}

\medskip
\noindent
\textbf{Step 3: explicit formula for interior weights.}
The row equation $E_i$ can be rewritten as
\[
\left(\sum_{j<i} Z_{ji}w_j+w_i\right)
+
\left(w_i+\sum_{j>i} Z_{ij}w_j\right)
-w_i
=1.
\]
Substituting \eqref{eq:Ai} and \eqref{eq:Bi} yields
\[
\frac{1}{1+a_i}+\frac{1}{1+a_{i-1}}-w_i=1.
\]
Hence for $i=2,\dots,k-1$,
\begin{equation}
\label{eq:wi-interior}
 w_i=\frac{1}{1+a_{i-1}}+\frac{1}{1+a_i}-1.
\end{equation}

\medskip
\noindent
\textbf{Step 4: summation.}
Using \eqref{eq:w1}, \eqref{eq:wk}, and \eqref{eq:wi-interior},
\begin{align*}
\SP(X)
&=\sum_{i=1}^k w_i \\
&=\frac{1}{1+a_1}
+\sum_{i=2}^{k-1}\left(\frac{1}{1+a_{i-1}}+\frac{1}{1+a_i}-1\right)
+\frac{1}{1+a_{k-1}} \\
&=2\sum_{r=1}^{k-1}\frac{1}{1+a_r}-(k-2).
\end{align*}
Now
\[
2\frac{1}{1+a_r}-1=\frac{1-a_r}{1+a_r},
\]
so
\[
\SP(X)
=1+\sum_{r=1}^{k-1}\frac{1-a_r}{1+a_r}.
\]
Finally, with $a_r=e^{-qg_r}$,
\[
\frac{1-a_r}{1+a_r}
=\frac{1-e^{-qg_r}}{1+e^{-qg_r}}
=\tanh\!\left(\frac{qg_r}{2}\right).
\]
This proves the claimed formula.
\end{proof}

The theorem shows that the global matrix quantity $\1^\top Z^{-1}\1$ collapses,
on the line, to an additive objective over consecutive selected gaps.
This is the structural fact behind the dynamic programming recursion.

\section{Fixed-cardinality subset selection and dynamic programming}\label{sec:dp}

Let
\[
Y=\{y_1,\dots,y_n\}\subseteq \R,
\qquad
y_1<\cdots <y_n,
\]
be an ordered candidate set. For a fixed cardinality $k\in\{1,\dots,n\}$, we
want to maximize the Solow--Polasky diversity over all $k$-element subsets
of $Y$.

Take an increasing index set
\[
I=\{i_1,\dots,i_k\},
\qquad
1\le i_1<\cdots <i_k\le n,
\]
and write
\[
Y_I:=\{y_{i_1},\dots,y_{i_k}\}.
\]
By Theorem~\ref{thm:gap-formula},
\begin{equation}
\label{eq:subset-gap-formula}
\SP(Y_I)
=1+\sum_{r=1}^{k-1}
\tanh\!\left(\frac{q\,(y_{i_{r+1}}-y_{i_r})}{2}\right).
\end{equation}
Hence only the distances between consecutive \emph{selected} points matter.
Define the edge weight
\begin{equation}
\label{eq:phi-def}
\phi(i,j):=\tanh\!\left(\frac{q\,(y_j-y_i)}{2}\right),
\qquad 1\le i<j\le n.
\end{equation}
Then maximizing $\SP(Y_I)$ is equivalent to maximizing
\[
\sum_{r=1}^{k-1} \phi(i_r,i_{r+1}).
\]

\begin{theorem}[Bellman recursion]
\label{thm:bellman}
For $m\in\{1,\dots,k\}$ and $j\in\{1,\dots,n\}$ define $F(m,j)$ as the maximum
value of
\[
\sum_{r=1}^{m-1}\phi(i_r,i_{r+1})
\]
over all strictly increasing index sequences
\[
1\le i_1<\cdots <i_m=j.
\]
Then
\[
F(1,j)=0
\qquad (j=1,\dots,n),
\]
and for $m\ge 2$,
\begin{equation}
\label{eq:bellman}
F(m,j)=\max_{1\le i<j}\bigl\{F(m-1,i)+\phi(i,j)\bigr\}.
\end{equation}
Moreover,
\begin{equation}
\label{eq:opt-value}
\max_{|I|=k}\SP(Y_I)
=1+\max_{1\le j\le n}F(k,j).
\end{equation}
Thus fixed-cardinality Solow--Polasky subset selection on the ordered line is
solvable in $O(kn^2)$ time and $O(kn)$ memory.
\end{theorem}

\begin{proof}
The base case $F(1,j)=0$ is immediate, because a one-point subset contains no
selected gap and hence contributes no edge weight.

Now fix $m\ge 2$ and $j\in\{1,\dots,n\}$. Consider any feasible increasing chain
\[
1\le i_1<\cdots <i_{m-1}<i_m=j.
\]
Its value is
\[
\sum_{r=1}^{m-1}\phi(i_r,i_{r+1})
=
\left(\sum_{r=1}^{m-2}\phi(i_r,i_{r+1})\right)+\phi(i_{m-1},j).
\]
If the prefix $i_1<\cdots <i_{m-1}$ were not optimal among all $(m-1)$-chains
ending at $i_{m-1}$, then replacing it by a better such prefix would increase
the value of the full $m$-chain ending at $j$. Therefore every optimal
$m$-chain ending at $j$ is obtained by choosing some predecessor $i<j$, taking
an optimal $(m-1)$-chain ending at $i$, and then appending $j$. This is
exactly the recursion \eqref{eq:bellman}.

Finally, every $k$-element subset corresponds to a unique increasing index chain
ending at some $j$, and by \eqref{eq:subset-gap-formula} its full
Solow--Polasky value is
\[
1+\sum_{r=1}^{k-1}\phi(i_r,i_{r+1}).
\]
Taking the maximum over all terminal indices $j$ gives \eqref{eq:opt-value}.
The stated complexity follows from evaluating \eqref{eq:bellman} for all
$m=2,\dots,k$ and all $j=1,\dots,n$.
\end{proof}

To recover an optimal subset itself, one stores for each state $(m,j)$ a
predecessor index attaining the maximum in \eqref{eq:bellman} and backtracks
from an index $j^*$ attaining $\max_j F(k,j)$.

For instance, if
\[
Y=\left\{0,\frac14,\frac12,\frac23,1\right\}
\qquad\text{and}\qquad
k=3,
\]
then the candidate chains ending at $y_5=1$ are
\[
(1,2,5),\ (1,3,5),\ (1,4,5),\ (2,4,5),\ \dots
\]
and the dynamic program compares the values
\[
F(2,1)+\phi(1,5),\quad
F(2,2)+\phi(2,5),\quad
F(2,3)+\phi(3,5),\quad
F(2,4)+\phi(4,5).
\]
Because the objective is increasing and concave in the selected gaps, the
best $3$-subset is $\{0,\tfrac12,1\}$: it uses two equal selected gaps of
length $\tfrac12$, giving the value
\[
\SP(\{0,\tfrac12,1\})=1+2\tanh\!\left(\frac{q}{4}\right).
\]

\begin{remark}[MPD objective]
The same ordered-chain structure also makes the fixed-cardinality MPD
objective exact. For
\(
I=\{i_1<\cdots<i_k\},
\)
define
\[
\mathrm{MPD}(Y_I):=\min_{1\le r<s\le k} d(y_{i_r},y_{i_s}).
\]
On the ordered line, and hence on every ordered $\ell_1$ staircase that reduces
to a line metric, the closest selected pair is necessarily consecutive in the
chosen order. Therefore
\[
\mathrm{MPD}(Y_I)=\min_{r=1,
\dots,k-1}\bigl(y_{i_{r+1}}-y_{i_r}\bigr).
\]
This gives the exact Bellman recursion
\begin{equation}
\label{eq:maxmin-bellman}
B(1,j)=+\infty,
\qquad
B(m,j)=\max_{1\le i<j}\min\{B(m-1,i),\ y_j-y_i\},
\end{equation}
for $m\ge2$, followed by $\max_j B(k,j)$. The principle of optimality is precise
here for the same reason as above: once the predecessor $i$ of $j$ is fixed, the
best prefix is the one that maximizes its bottleneck value. In contrast to the
Solow--Polasky expression in \eqref{eq:subset-gap-formula}, however, this
objective is not additive over gaps. Its concatenation law is the bottleneck
law
\[
\mathrm{MPD}(\text{prefix followed by }j)
=
\min\{\mathrm{MPD}(\text{prefix}),\ y_j-y_i\},
\]
which is subadditive rather than additive and leads to a bottleneck rather than a
max--plus dynamic program.

It can also be regarded as the
$s\to\infty$ limiting case of Riesz $s$-energy subset selection, since for any
fixed selected set with positive pairwise distances,
\(
\left(\sum_{r<s} d(y_{i_r},y_{i_s})^{-s}\right)^{-1/s}
\longrightarrow
\min_{r<s}d(y_{i_r},y_{i_s}).
\)
This explains why ordered-point dynamic-programming ideas for Riesz energy are
closely related to, but not identical with, the exact bottleneck recursion above;
see also Emmerich~\cite{Emmerich2025RieszDP}.
\end{remark}

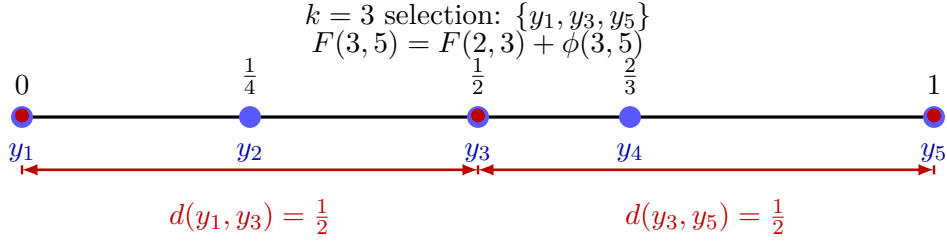
\begin{figure}[t]
\centering
\begin{tikzpicture}[x=12cm,y=1cm,>=latex]
  \draw[line width=1.2pt] (0,0) -- (1,0);

  \foreach \x/\lbl/\name in {
    0/{0}/{$y_1$},
    0.25/{\frac14}/{$y_2$},
    0.5/{\frac12}/{$y_3$},
    0.6667/{\frac23}/{$y_4$},
    1/{1}/{$y_5$}
  }{
    \fill[blue!65] (\x,0) circle (3.8pt);
    \draw[blue!65,line width=0.6pt] (\x,0) circle (3.8pt);
    \node[below=6pt,blue!70!black] at (\x,0) {\name};
    \node[above=5pt] at (\x,0) {$\lbl$};
  }

  \foreach \x in {0,0.5,1}{
    \fill[red!75!black] (\x,0.02) circle (2.6pt);
  }

  \node[align=center,anchor=south] at (0.5,0.63) {$k=3$ selection: $\{y_1,y_3,y_5\}$\\[-1mm]
  $F(3,5)=F(2,3)+\phi(3,5)$};

  \draw[red!75!black,line width=0.9pt] (0,-0.64) -- (0,-0.76);
  \draw[red!75!black,line width=0.9pt] (0.5,-0.64) -- (0.5,-0.76);
  \draw[red!75!black,<->,line width=1.0pt] (0,-0.70) -- (0.5,-0.70)
    node[midway,below=7pt] {$d(y_1,y_3)=\frac{1}{2}$};
  \draw[red!75!black,line width=0.9pt] (1,-0.64) -- (1,-0.76);
  \draw[red!75!black,<->,line width=1.0pt] (0.5,-0.70) -- (1,-0.70)
    node[midway,below=7pt] {$d(y_3,y_5)=\frac{1}{2}$};
\end{tikzpicture}
\caption{A small dynamic-programming example for $Y=\{0,\frac14,\frac12,\frac23,1\}$ and $k=3$. The five candidate points are shown in blue, while the optimal $3$-point subset is shown with red center. Since
\(
\SP(Y_I)=1+\sum_r \phi(i_r,i_{r+1}),
\)
the dynamic program only has to track the best predecessor of the current terminal point.}
\label{fig:dp-example}
\end{figure}

\section{Ordered \texorpdfstring{$\ell_1$}{l1} Pareto fronts in the plane}\label{sec:l1-pareto}

The line result extends exactly to a common two-dimensional Pareto-front setting
when distances are measured in the $\ell_1$ norm.

\begin{theorem}[Ordered $\ell_1$ planar reduction]
\label{thm:l1-planar-reduction}
Let
\[
p_i=(u_i,v_i)\in\R^2,
\qquad i=1,\dots,n,
\]
with
\[
u_1<\cdots <u_n
\qquad\text{and}\qquad
v_1>\cdots >v_n.
\]
Thus the points form an ordered monotone chain, for example a mutually
non-dominated set for a bicriteria minimization problem. Define
\[
t_i:=u_i-v_i.
\]
Then $t_1<\cdots<t_n$ and, for every $i<j$,
\[
\|p_j-p_i\|_1=t_j-t_i.
\]
Consequently, for the exponential kernel
\[
Z_{ij}=e^{-q\|p_i-p_j\|_1},
\]
the similarity matrix of the ordered planar set is exactly the similarity matrix
of the ordered line set $\{t_1,\dots,t_n\}$. Hence all formulas from
Theorem~\ref{thm:gap-formula} and the dynamic program from
Theorem~\ref{thm:bellman} apply verbatim after replacing $y_i$ by $t_i$.
In particular, fixed-cardinality Solow--Polasky subset selection for such
ordered $\ell_1$ point sets is solvable exactly in $O(kn^2)$ time and $O(kn)$
memory.
\end{theorem}

\begin{proof}
For $i<j$, monotonicity of the two coordinates gives
\[
\|p_j-p_i\|_1
=|u_j-u_i|+|v_j-v_i|
=(u_j-u_i)+(v_i-v_j).
\]
Since $t_i=u_i-v_i$, this equals
\[
(u_j-v_j)-(u_i-v_i)=t_j-t_i.
\]
Moreover, $t_j-t_i=(u_j-u_i)+(v_i-v_j)>0$ for $i<j$, so the $t_i$ are strictly
increasing. Therefore the whole $\ell_1$ distance matrix of the planar chain is
identical to the distance matrix of the ordered line set $\{t_1,\dots,t_n\}$.
The Solow--Polasky similarity matrix and objective are functions only of this
distance matrix, so the consecutive-gap formula and Bellman recursion from
Theorems~\ref{thm:gap-formula} and~\ref{thm:bellman} apply without further
modification. The complexity bound is inherited from Theorem~\ref{thm:bellman}.
\end{proof}

The following examples illustrate the planar reduction before the higher-dimensional
staircase formulation is introduced. The first is deliberately small and shows
how equal selected $\ell_1$ gaps arise after the map $t_i=u_i-v_i$. The second
uses a larger sampled Pareto-front approximation and shows the same mechanism in
a less symmetric situation.

\subsection*{Worked Pareto-front example in $\ell_1$}
Consider the ordered Pareto-front approximation
\[
p_1=(0,5),\quad p_2=(2,3),\quad p_3=\left(\frac52,\frac52\right),\quad p_4=\left(4,\frac12\right),\quad p_5=(5,0),
\]
for a bicriteria minimization problem. The points are mutually non-dominated and
already ordered by increasing first objective and decreasing second objective.
With the notation of Theorem~\ref{thm:l1-planar-reduction},
\[
t_i=u_i-v_i,
\]
one obtains the line representation
\[
t_1=-5,\qquad t_2=-1,\qquad t_3=0,\qquad t_4=\frac72,\qquad t_5=5.
\]
Therefore fixed-cardinality subset selection on this planar front is exactly the
same as the line problem for the ordered set
\[
\{-5,-1,0,\tfrac72,5\}.
\]

For the concrete computation shown here, set $q=1$ and select $k=3$ points.
A Python implementation of the Bellman recursion returns the subset
\[
\{p_1,p_3,p_5\},
\]
with selected $\ell_1$ gaps
\[
\|p_3-p_1\|_1=5,\qquad \|p_5-p_3\|_1=5.
\]
Hence
\[
\SP(\{p_1,p_3,p_5\})
=1+2\tanh\!\left(\frac52\right)
\approx 2.9732.
\]
This example makes the uniform-spacing principle especially transparent: among all
$3$-subsets of these five ordered points, the selected subset induces two equal
consecutive $\ell_1$ gaps, and therefore maximizes the sum of the concave terms
$\tanh(qd/2)$ under the fixed total gap $\|p_5-p_1\|_1=10$.

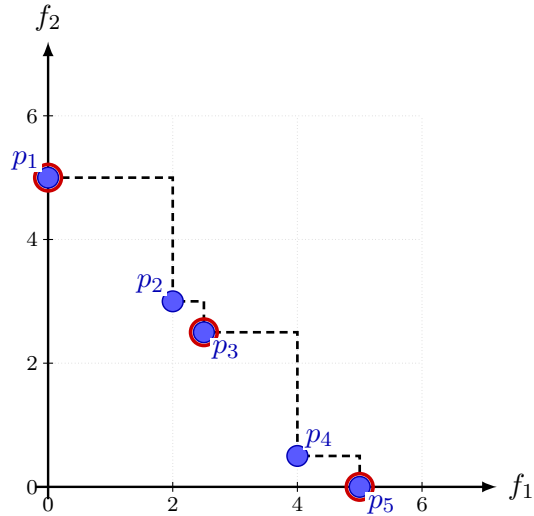
\begin{figure}[h]
\centering
\begin{tikzpicture}[x=0.82cm,y=0.82cm,>=latex]
  \draw[gray!35,densely dotted,line width=0.25pt] (0,0) grid[step=2] (6,6);

  \draw[->,line width=0.9pt] (-0.2,0) -- (7.2,0) node[right] {$f_1$};
  \draw[->,line width=0.9pt] (0,-0.2) -- (0,7.2) node[above] {$f_2$};
  \foreach \t in {0,2,4,6}{
    \draw[line width=0.4pt] (\t,0.08) -- (\t,-0.08);
    \draw[line width=0.4pt] (0.08,\t) -- (-0.08,\t);
    \node[below] at (\t,0) {\scriptsize $\t$};
    \node[left] at (0,\t) {\scriptsize $\t$};
  }

  \draw[densely dashed,line width=1.0pt]
    (0,5) -- (2,5) -- (2,3) -- (2.5,3) -- (2.5,2.5) -- (4,2.5) -- (4,0.5) -- (5,0.5) -- (5,0);

  \foreach \x/\y in {0/5,2/3,2.5/2.5,4/0.5,5/0}{
    \fill[blue!65] (\x,\y) circle (3.9pt);
    \draw[blue!70!black,line width=0.55pt] (\x,\y) circle (3.9pt);
  }

  \foreach \x/\y in {0/5,2.5/2.5,5/0}{
    \draw[red!80!black,line width=1.35pt] (\x,\y) circle (5.0pt);
  }

  \node[above left=2pt,fill=white,inner sep=1pt,text=blue!70!black] at (0,5) {$p_1$};
  \node[above left=2pt,fill=white,inner sep=1pt,text=blue!70!black] at (2,3) {$p_2$};
  \node[below right=2pt,fill=white,inner sep=1pt,text=blue!70!black] at (2.5,2.5) {$p_3$};
  \node[above right=2pt,fill=white,inner sep=1pt,text=blue!70!black] at (4,0.5) {$p_4$};
  \node[below right=2pt,fill=white,inner sep=1pt,text=blue!70!black] at (5,0) {$p_5$};
\end{tikzpicture}
\caption{An ordered Pareto-front example in $\R^2$ with the $\ell_1$ norm. The
five candidate points are shown in blue, and the optimal $3$-subset for $q=1$ is
indicated by the red boundary circles, namely $\{p_1,p_3,p_5\}$. The dashed staircase
runs from $(0,5)$ to $(5,0)$ by moving first horizontally and then vertically through
the intermediate abscissae; its total length is the $\ell_1$-distance between the
endpoints, $\|p_5-p_1\|_1=10$. The selected subset realizes the two equal consecutive
gaps $\|p_3-p_1\|_1=\|p_5-p_3\|_1=5$. A very light dotted grid is shown at every second
integer level to guide the eye without suggesting Euclidean interpolation.}
\label{fig:pareto-l1-example}
\end{figure}

A compact Python implementation of the Bellman recursion for this example is given in Listing~\ref{lst:pareto-dp} at the end of this section.

\subsection*{Random Pareto-front example on $f_2=1-f_1^2$}
As a larger illustration, let
\[
p_i=(x_i,1-x_i^2),\qquad i=1,\dots,20,
\]
where the abscissae $x_i$ are sampled uniformly at random from $[0,1]$ and then
sorted increasingly. The reproducible seed is $10$, so the points form a
strictly ordered Pareto-front approximation on the concave curve $f_2=1-f_1^2$.
Because the front is ordered and monotone, Theorem~\ref{thm:l1-planar-reduction} applies and the
$\ell_1$-based Solow--Polasky subset selection problem is exactly equivalent to a
line problem.

For the concrete computation, again set $q=1$ and choose fixed cardinality
$k=6$. Running the Bellman recursion on the resulting ordered $20$-point set
selects the indices
\(
\{1,6,10,15,18,20\},
\)
that is,
\[
\begin{aligned}
I^*=\{&(0.0446,0.9980),\ (0.3278,0.8926),\ (0.5207,0.7289),\\
      &(0.6750,0.5444),\ (0.8602,0.2601),\ (0.9966,0.0069)\}
\end{aligned}
\]
and the corresponding value is
\(
\SP(I^*)\approx 1.9590.
\)
The selected points are spread rather evenly in the induced line coordinate
$t_i=f_{1,i}-f_{2,i}=x_i^2+x_i-1$, which is exactly what the gap-additive
objective favors.

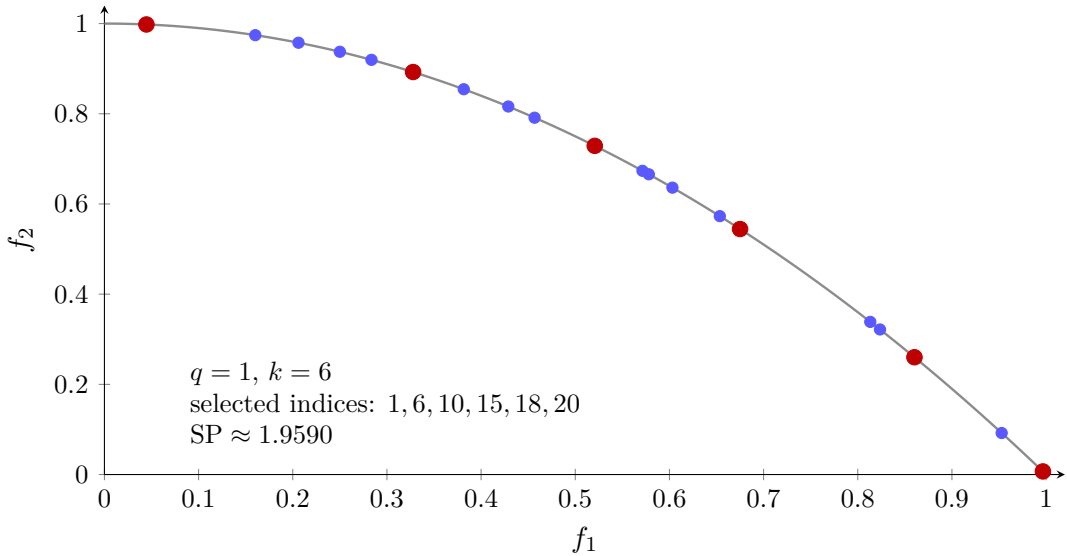
\begin{figure}[H]
\centering
\begin{tikzpicture}
\begin{axis}[
  width=0.9\textwidth,
  height=7.8cm,
  xmin=0, xmax=1.02,
  ymin=0, ymax=1.04,
  axis lines=left,
  xlabel={$f_1$},
  ylabel={$f_2$},
  enlargelimits=false,
  tick label style={font=\small},
  label style={font=\normalsize},
  clip=false
]

\addplot[domain=0:1,samples=201,smooth,black!45,line width=0.9pt] {1-x^2};

\addplot[
  only marks,
  mark=*,
  mark size=2.1pt,
  blue!65
] coordinates {
(0.0446,0.9980)
(0.1602,0.9743)
(0.2061,0.9575)
(0.2500,0.9375)
(0.2836,0.9196)
(0.3278,0.8926)
(0.3816,0.8544)
(0.4289,0.8161)
(0.4568,0.7913)
(0.5207,0.7289)
(0.5714,0.6735)
(0.5781,0.6658)
(0.6032,0.6362)
(0.6535,0.5730)
(0.6750,0.5444)
(0.8133,0.3385)
(0.8236,0.3217)
(0.8602,0.2601)
(0.9528,0.0921)
(0.9966,0.0069)
};

\addplot[
  only marks,
  mark=*,
  mark size=2.9pt,
  red!75!black
] coordinates {
(0.0446,0.9980)
(0.3278,0.8926)
(0.5207,0.7289)
(0.6750,0.5444)
(0.8602,0.2601)
(0.9966,0.0069)
};

\node[anchor=west,align=left,font=\small] at (axis cs:0.08,0.16)
{$q=1$, $k=6$\\selected indices: $1,6,10,15,18,20$\\$\SP\approx 1.9590$};

\end{axis}
\end{tikzpicture}
\caption{A random Pareto-front approximation with $20$ sampled points on the
curve $f_2=1-f_1^2$ (blue), generated with seed $10$. The exact $\ell_1$
Solow--Polasky optimum for $q=1$ and fixed cardinality $k=6$ is shown in red.
Because the front is ordered and monotone, the Bellman recursion from
Theorem~\ref{thm:bellman} applies without change after mapping the points
to their induced line coordinates $t_i=f_{1,i}-f_{2,i}$.}
\label{fig:random-pareto-pgf}
\end{figure}
These examples emphasize that no Euclidean interpolation is used: the relevant
geometry is the ordered $\ell_1$ geometry along the staircase, or equivalently
the induced line coordinate. The next section abstracts this observation from
monotone planar Pareto fronts to signed coordinatewise monotone chains in
arbitrary dimension.
\clearpage
\noindent For reproducibility, Listing~\ref{lst:pareto-dp} gives the compact dynamic program used for the worked five-point example. It is placed on a separate page so that the examples and the code can be read independently.

\vspace{1ex}
\begin{lstlisting}[basicstyle=\ttfamily\footnotesize,caption={Python dynamic programming for the Pareto-front example with $q=1$ and $k=3$.},label={lst:pareto-dp}]
import math

points = [(0.0, 5.0), (2.0, 3.0), (2.5, 2.5), (4.0, 0.5), (5.0, 0.0)]
q = 1.0
k = 3


def l1_distance(p, r):
    return abs(p[0] - r[0]) + abs(p[1] - r[1])


def solve_fixed_cardinality(points, k, q=1.0):
    n = len(points)
    F = [[float("-inf")] * n for _ in range(k + 1)]
    pred = [[None] * n for _ in range(k + 1)]

    for j in range(n):
        F[1][j] = 0.0

    for m in range(2, k + 1):
        for j in range(n):
            best_value = float("-inf")
            best_pred = None
            for i in range(j):
                d = l1_distance(points[i], points[j])
                cand = F[m - 1][i] + math.tanh(q * d / 2.0)
                if cand > best_value + 1e-15:
                    best_value = cand
                    best_pred = i
            F[m][j] = best_value
            pred[m][j] = best_pred

    best_last = max(range(n), key=lambda j: F[k][j])
    value = 1.0 + F[k][best_last]

    subset = []
    m, j = k, best_last
    while m >= 1:
        subset.append(j)
        j = pred[m][j]
        m -= 1
    subset.reverse()
    return value, [i + 1 for i in subset]


value, subset = solve_fixed_cardinality(points, k, q)
print("selected indices:", subset)
print("SP value:", value)
\end{lstlisting}

\clearpage

\section{Monotone \texorpdfstring{$\ell_1$}{l1} staircases in higher dimensions}
\label{sec:l1-staircases}

The planar reduction in Theorem~\ref{thm:l1-planar-reduction} is a special case
of a simple higher-dimensional principle. The essential requirement is not that
the points lie in the plane, but that they can be traversed without backtracking
in any coordinate direction.

\begin{definition}[Signed coordinatewise monotone $\ell_1$ staircase]
Let
\[
S=\{p_1,\dots,p_n\}\subset\R^d.
\]
We say that $S$ admits an $\ell_1$-staircase ordering if there exists a
permutation $\pi$ of $\{1,\dots,n\}$ and a sign vector
\[
\sigma=(\sigma_1,\dots,\sigma_d)\in\{-1,+1\}^d
\]
such that, with $q_i:=p_{\pi(i)}$,
\[
\sigma_\ell q_{1\ell}\le \sigma_\ell q_{2\ell}\le\cdots\le
\sigma_\ell q_{n\ell}
\qquad\text{for every }\ell=1,\dots,d.
\]
Equivalently, after possibly reversing some coordinate axes, the same ordering
sorts all coordinate projections. The associated $\ell_1$ staircase is the
axis-parallel monotone polyline obtained by connecting consecutive points
$q_i,q_{i+1}$ through coordinate-parallel moves.
\end{definition}

For $\sigma=(1,\dots,1)$ this is an ascending staircase. The two-dimensional
Pareto-front case from Theorem~\ref{thm:l1-planar-reduction} corresponds to
$d=2$ and $\sigma=(1,-1)$.

\begin{proposition}[Characterization of ordered $\ell_1$ line reductions]
\label{prop:l1-staircase-characterization}
Let $q_1,\dots,q_n\in\R^d$ be given in a fixed order and let
$d_1(x,y):=\|x-y\|_1$. The following are equivalent.
\begin{enumerate}[label=\emph{(\roman*)}]
\item The ordered points form a signed coordinatewise monotone chain.
\item For every $1\le i<j\le n$,
\[
 d_1(q_i,q_j)=\sum_{r=i}^{j-1}d_1(q_r,q_{r+1}).
\]
\item There are numbers $t_1<\cdots<t_n$ such that, for every $i<j$,
\[
 d_1(q_i,q_j)=t_j-t_i.
\]
\end{enumerate}
When these conditions hold, the similarity matrix
$Z_{ij}=e^{-q d_1(q_i,q_j)}$ is exactly the similarity matrix of the ordered line
set $\{t_1,\dots,t_n\}$. Hence the consecutive-gap formula from
Theorem~\ref{thm:gap-formula} and the Bellman recursion from
Theorem~\ref{thm:bellman} apply without modification. The bottleneck recursion
\eqref{eq:maxmin-bellman} for MPD applies in the same way after
replacing $y_j-y_i$ by $t_j-t_i$. In this sense, coordinatewise monotone
$\ell_1$ staircases are precisely the ordered $\ell_1$ point sets for which
these line-metric dynamic programs apply as direct reductions.
\end{proposition}

\begin{proof}
Assume first that the points are signed coordinatewise monotone, and choose the
corresponding sign vector $\sigma$. Define
\[
t_i:=\sum_{\ell=1}^d \sigma_\ell q_{i\ell}.
\]
Then, for $i<j$, each signed coordinate difference
$\sigma_\ell(q_{j\ell}-q_{i\ell})$ is nonnegative. Therefore
\[
\|q_j-q_i\|_1
=\sum_{\ell=1}^d |q_{j\ell}-q_{i\ell}|
=\sum_{\ell=1}^d \sigma_\ell(q_{j\ell}-q_{i\ell})
=t_j-t_i.
\]
This proves the line representation, and the additivity over consecutive gaps
follows immediately.

Conversely, suppose that the additivity condition in \emph{(ii)} holds. Write
$\Delta_{r\ell}:=q_{r+1,\ell}-q_{r\ell}$. Applying the triangle inequality in
each coordinate to the pair $q_1,q_n$ gives
\[
\sum_{\ell=1}^d\left|\sum_{r=1}^{n-1}\Delta_{r\ell}\right|
\le
\sum_{\ell=1}^d\sum_{r=1}^{n-1}|\Delta_{r\ell}|.
\]
Condition \emph{(ii)} says that equality holds. Hence equality must hold in each
coordinate separately, which is possible only if, for each fixed coordinate
$\ell$, all nonzero increments $\Delta_{r\ell}$ have the same sign. Thus every
coordinate sequence $q_{1\ell},\dots,q_{n\ell}$ is monotone, after possibly
reversing that coordinate. This is exactly \emph{(i)}. Finally,
\emph{(iii)} implies \emph{(ii)} by telescoping with
$t_{r+1}-t_r=d_1(q_r,q_{r+1})$, and the dynamic-programming statement follows
from Theorems~\ref{thm:gap-formula} and~\ref{thm:bellman}.
\end{proof}

\begin{remark}
The proposition should be read as a statement about this exact reduction, not as
a claim that no other special metric structures can admit dynamic programs. It
says that, inside ordered finite $\ell_1$ point sets, the absence of coordinate
backtracking is exactly what turns the $\ell_1$-distance matrix into the
distance matrix of an ordered set on the real line.
\end{remark}

\subsection*{A three-dimensional staircase example}
Consider the ascending chain
\[
q_1=(0,0,0),\qquad
q_2=(1,1,2),\qquad
q_3=(2,3,3),\qquad
q_4=(4,5,6).
\]
The same ordering sorts all three coordinate projections:
\[
0\le 1\le 2\le 4,
\qquad
0\le 1\le 3\le 5,
\qquad
0\le 2\le 3\le 6.
\]
With $\sigma=(1,1,1)$, the induced line coordinates are
\[
t_1=0,
\qquad
 t_2=4,
\qquad
 t_3=8,
\qquad
 t_4=15.
\]
Thus the consecutive $\ell_1$ gaps are $4,4,7$, and the total staircase length
from $q_1$ to $q_4$ is $15$. Figure~\ref{fig:3d-l1-staircase} illustrates the
axis-parallel path and two light-gray coordinate-plane projections.

\begin{figure}[h]
\centering
\begin{tikzpicture}[
  x={(1.08cm,0cm)},
  y={(0.84cm,0.35cm)},
  z={(0cm,0.97cm)},
  >=latex,
  line join=round,
  line cap=round
]

  \draw[->,line width=1.0pt] (0,0,0) -- (7.2,0,0)
    node[below,font=\bfseries] {$f_1$};
  \draw[->,line width=1.0pt] (0,0,0) -- (0,7.0,0)
    node[right,font=\bfseries] {$f_2$};
  \draw[->,line width=1.0pt] (0,0,0) -- (0,0,7.8)
    node[left,font=\bfseries] {$f_3$};

  \foreach \t in {0,1,2,3,4,5,6,7}{
    \draw[line width=0.35pt] (\t,0,0) -- +(0,0,-0.08);
    \node[below,font=\scriptsize\bfseries] at (\t,0,0) {\t};
  }

  \foreach \t in {1,2,3,4,5,6}{
    \draw[line width=0.35pt] (0,\t,0) -- +(0.0,-0.18);
    \node[left,font=\scriptsize\bfseries] at (0,\t,0.1) {\t};
  }

  \foreach \t in {0,1,2,3,4,5,6,7}{
    \draw[line width=0.35pt] (0,0,\t) -- +(-0.08,0,0);
    \node[left,font=\scriptsize\bfseries] at (0,0,\t) {\t};
  }

  \draw[densely dashed,line width=1.15pt,gray!55]
    (0,0,0) -- (1,0,0) -- (1,1,0) -- (2,1,0) -- (2,3,0) -- (4,3,0) -- (4,5,0);

  \foreach \p in {(0,0,0),(1,1,0),(2,3,0),(4,5,0)}{
    \fill[gray!30] \p circle (2.8pt);
    \draw[gray!55,line width=0.45pt] \p circle (2.8pt);
  }

  \draw[densely dotted,line width=0.8pt,gray!50] (1,1,2) -- (1,1,0);
  \draw[densely dotted,line width=0.8pt,gray!50] (2,3,3) -- (2,3,0);
  \draw[densely dotted,line width=0.8pt,gray!50] (4,5,6) -- (4,5,0);

  \draw[densely dashed,line width=1.15pt,gray!55]
    (0,0,0) -- (0,1,0) -- (0,1,2) -- (0,3,2) -- (0,3,3) -- (0,5,3) -- (0,5,6);

  \foreach \p in {(0,0,0),(0,1,2),(0,3,3),(0,5,6)}{
    \fill[gray!30] \p circle (2.8pt);
    \draw[gray!55,line width=0.45pt] \p circle (2.8pt);
  }

  \draw[densely dotted,line width=0.8pt,gray!50] (1,1,2) -- (0,1,2);
  \draw[densely dotted,line width=0.8pt,gray!50] (2,3,3) -- (0,3,3);
  \draw[densely dotted,line width=0.8pt,gray!50] (4,5,6) -- (0,5,6);

  \draw[densely dashed,line width=2.0pt,blue!75!black]
    (0,0,0) -- (1,0,0);
  \draw[densely dashed,line width=2.0pt,blue!75!black]
    (1,1,2) -- (2,1,2);
  \draw[densely dashed,line width=2.0pt,blue!75!black]
    (2,3,3) -- (4,3,3);

  \draw[densely dashed,line width=2.0pt,green!50!black]
    (1,0,0) -- (1,1,0);
  \draw[densely dashed,line width=2.0pt,green!50!black]
    (2,1,2) -- (2,3,2);
  \draw[densely dashed,line width=2.0pt,green!50!black]
    (4,3,3) -- (4,5,3);

  \draw[densely dashed,line width=2.0pt,yellow!70!black]
    (1,1,0) -- (1,1,2);
  \draw[densely dashed,line width=2.0pt,yellow!70!black]
    (2,3,2) -- (2,3,3);
  \draw[densely dashed,line width=2.0pt,yellow!70!black]
    (4,5,3) -- (4,5,6);

  \foreach \p/\pos/\lbl in {
    {(0,0,0)}/{below left}/{$\mathbf{q}_1$},
    {(1,1,2)}/{above right}/{$\mathbf{q}_2$},
    {(2,3,3)}/{above right}/{$\mathbf{q}_3$},
    {(4,5,6)}/{above right}/{$\mathbf{q}_4$}
  }{
    \fill[blue!65] \p circle (4.0pt);
    \draw[blue!70!black,line width=0.8pt] \p circle (4.0pt);
    \node[\pos=2pt,fill=white,inner sep=1pt,font=\bfseries] at \p {\lbl};
  }

  \node[anchor=west,font=\bfseries,align=left] at (2.2,5.9,0.35)
    {$\ell_1$ staircase length $=4+5+6=15$};

\end{tikzpicture}
\caption{A three-dimensional ascending $\ell_1$ staircase with colored coordinate
directions. Blue indicates motion in the $f_1$-direction, dark green in
the $f_2$-direction, and dark yellow in the $f_3$-direction.
The light-gray auxiliary polylines and points show the projections onto the
$f_1$-$f_2$ plane and the $f_2$-$f_3$ plane.}
\label{fig:3d-l1-staircase}
\end{figure}
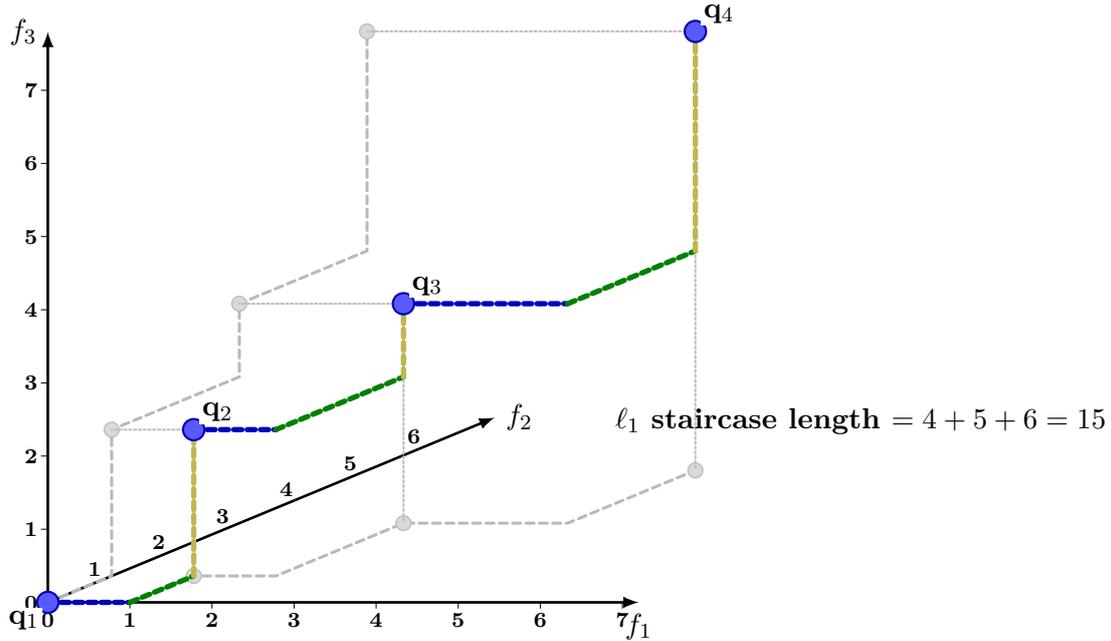

\section{Conclusion and future work}\label{sec:conclusion}

The main conclusion is that fixed-cardinality Solow--Polasky diversity subset
selection, although hard in general metric spaces and even in the Euclidean
plane~\cite{EmmerichPereverdievaDeutz2026Metric,EmmerichPereverdievaDeutz2026Plane},
has an exact polynomial-time solution on ordered one-dimensional sets, on
monotone biobjective Pareto fronts under the $\ell_1$ metric, and on the
higher-dimensional coordinatewise monotone $\ell_1$ staircases characterized in
Proposition~\ref{prop:l1-staircase-characterization}. The reason is that the
finite exponential-kernel objective becomes the magnitude of the scaled metric
space and, on any ordered $\ell_1$ chain without coordinate backtracking, this
value is an additive function of the consecutive selected gaps.

The ordered $\ell_1$ case studied here may thus be viewed as another tractable regime for exact Solow--Polasky diversity optimization. This complements, for example, the finite ultrametric case discussed by Leinster and Meckes~\cite{LeinsterMeckes2016}, where the metric satisfies\(d(x,z)\leq \max\{d(x,y),d(y,z)\}\).

Future work should investigate which tractable subclasses remain beyond these
coordinatewise monotone $\ell_1$ staircases. Natural directions include partially
ordered subsets in $\R^d$ that can be decomposed into a small number of monotone
chains, higher-dimensional Pareto-front structures that are not totally ordered
in all coordinate projections, and structural conditions under which exact
recursions, approximation schemes, or useful bounds can be obtained for
Solow--Polasky subset selection.

A related case of diversity-based subset selection is the problem of minimal Riesz $s$-energy subset selection in ordered point
sets. There, dynamic-programming ideas can exploit an ordering, but the
left--right recurrence considered by Emmerich is approximate and need not return
the exact optimum~\cite{Emmerich2025RieszDP} but becomes exact to the MPD limit case ($s\to\infty$), when it provides subadditive rather than
the additive transition used for Solow--Polasky diversity.

Another related question is the
representation of two-dimensional Pareto fronts of large cardinality by smaller
finite subsets. For the hypervolume indicator and the epsilon-indicator, this
setting has led to algorithms with asymptotically faster running times than the
straightforward left--right dynamic programming used here~\cite{BringmannFriedrichKlitzke2014,KuhnEtAl2016}.
In those cases, the optimized measures can be decomposed into sums of terms that
only depend on neighboring selected points. A possible direction is therefore to
ask whether some of these techniques can be translated to Solow--Polasky
diversity subset selection, either to speed up the present recursion or to
identify further exactly solvable ordered and Pareto-front cases.
A structurally similar dynamic-programming idea also appears in the selection of
representative points from Pareto fronts for the weighted hypervolume
indicator~\cite{AugerBaderBrockhoffZitzler2009}. This further supports the view
that ordered Pareto-front structure can make otherwise difficult subset
selection problems algorithmically tractable.

It is also noteworthy that, for hypervolume subset selection, the hard cases appear when moving from two-dimensional Pareto-front approximations to three and higher dimensions. In the anchored-box formulation, Bringmann, Cabello, and Emmerich prove NP-hardness already in dimension three, while the planar case remains solvable in polynomial time~\cite{BringmannCabelloEmmerich2017}. As with SP diversity subset selection, this contrast suggests a useful higher-dimensional research direction: to identify precisely which diversity indicators and geometric structures admit genuinely exact recursions, and which instead require  approximation methods, or fixed-parameter approaches, for example through tractable dynamic programming on bounded-tree-width structures.

\paragraph{Acknowledgements}
I am grateful to André H. Deutz, Ksenia Pereverdieva,  and Tom Leinster for helpful feedback and for drawing my attention to relevant work.

\end{document}